\documentclass{llncs}
\usepackage{amsmath,amssymb}
\usepackage{graphv1}
\newcommand{\thd}{\ensuremath{\mathrm{th}}}
\newcommand{\oth}[1]{\ensuremath{\overline{#1}}}

\newcounter{claimctr}[lemma]
\newtheorem{observation}{Observation}
\newenvironment{myclaim}{\refstepcounter{claimctr}\medskip\par\noindent\textit{Claim \theclaimctr.}}{\par}

\begin{document}
\pagestyle{plain}
	\title{The Lexicographic Method for the Threshold Cover Problem\thanks{A preliminary version of this paper, claiming the same result as proved in this work, appeared in the proceedings of the conference CALDAM 2020. But the proof in that version contains a serious error, and the algorithm mentioned in that paper may fail to produce a 2-threshold cover if the input graph contains a paraglider as an induced subgraph.}}
	\author{Mathew~C. Francis\and Dalu Jacob}
	\institute{Indian Statistical Institute, Chennai Centre\\\email{\{mathew,dalujacob\}@isichennai.res.in}}
	\maketitle
	\begin{abstract}
	Threshold graphs are a class of graphs that have many equivalent definitions and have applications in integer programming and set packing problems. A graph is said to have a threshold cover of size $k$ if its edges can be covered using $k$ threshold graphs. Chv\'atal and Hammer, in 1977, defined the \emph{threshold dimension} $\thd(G)$ of a graph $G$ to be the least integer $k$ such that $G$ has a threshold cover of size $k$ and observed that $\thd(G)\geq\chi(G^*)$, where $G^*$ is a suitably constructed auxiliary graph. Raschle and Simon~[\textit{Proceedings of the Twenty-seventh Annual ACM Symposium on Theory of Computing}, STOC '95, pages 650--661, 1995] proved that $\thd(G)=\chi(G^*)$ whenever $G^*$ is bipartite. We show how the lexicographic method of Hell and Huang can be used to obtain a completely new and, we believe, simpler proof for this result. For the case when $G$ is a split graph, our method yields a proof that is much shorter than the ones known in the literature.
	\end{abstract}
\section{Introduction}
	We consider only simple, undirected and finite graphs. We denote an edge between two vertices $u$ and $v$ of a graph by the two-element set $\{u,v\}$, which is usually abbreviated to just $uv$. Two edges $ab,cd$ in a graph $G$ are said to form an \emph{alternating 4-cycle} if $ad,bc\in E(\overline{G})$. A graph $G$ that does not contain any pair of edges that form an alternating 4-cycle is called a \emph{threshold graph}; or equivalently, $G$ is $(2K_2,P_4,C_4)$-free~\cite{chvatal1977aggregations}. A graph $G=(V,E)$ is said to be \emph{covered} by the graphs $H_1,H_2,\ldots, H_k$ if $E(G)=E(H_1)\cup E(H_2)\cup\cdots\cup E(H_k)$.
	\begin{definition}[Threshold cover and threshold dimension]
	A graph $G$ is said to have a \emph{threshold cover} of size $k$ if it can be covered by $k$ threshold graphs.
	The \emph{threshold dimension} of a graph $G$, denoted as $\thd(G)$, is defined to be the smallest integer $k$ such that $G$ has a threshold cover of size $k$.
	\end{definition}
	Mahadev and Peled~\cite{mahadev1995threshold} give a comprehensive survey of threshold graphs and their applications.
	
	Chv\'atal and Hammer~\cite{chvatal1977aggregations} showed that the fact that a graph $G$ has $\thd(G)\leq k$ is equivalent to the following: there exist $k$ linear inequalities on $|V(G)|$ variables such that the characteristic vector of a set $S\subseteq V(G)$ satisfies all the inequalities if and only if $S$ is an independent set of $G$ (see~\cite{raschle1995recognition} for details). They further defined the auxiliary graph $G^*$ corresponding to a graph $G$ as follows.

	\begin{definition}[Auxiliary graph]
	Given a graph $G$, the graph $G^*$ has vertex set $V(G^*)=E(G)$ and edge set $E(G^*)=\{\{ab,cd\}\colon ab,cd\in E(G)$ such that $ab,cd$ form an alternating 4-cycle in $G\}$.
	\end{definition}

	A graph $G$ and the auxiliary graph $G^*$ corresponding to it is shown in Figure~\ref{fig:auxgraph}.
	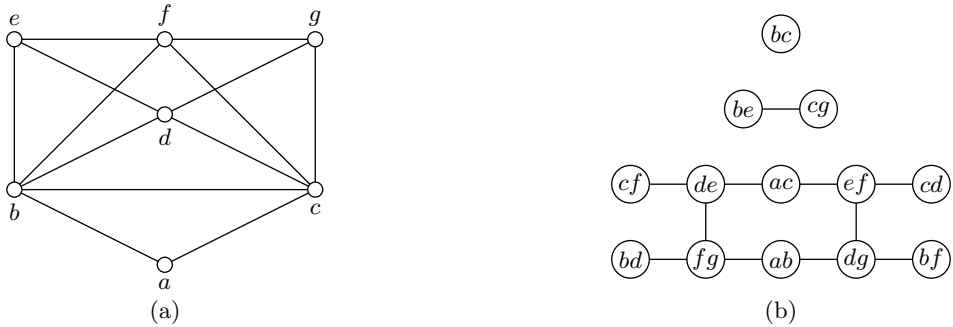
\begin{figure}
		\begin{tabular}{p{.49\textwidth}p{.49\textwidth}}
		\parbox{.49\textwidth}{
		\centering
		\renewcommand{\defradius}{0.1}
		\renewcommand{\vertexset}{(6,2,2),(7,0,3),(2,2,3),(4,0,1),(3,4,3),(5,4,1),(1,2,0)}
		\renewcommand{\edgeset}{(1,4),(1,5),(2,3),(2,4),(2,5),(2,7),(3,5),(3,6),(4,5),(4,6),(4,7),(5,6),(6,7)}
		\begin{tikzpicture}
		\drawgraph
		\node[below=2] at (\xy{1}) {$a$};
		\node[above=2] at (\xy{2}) {$f$};
		\node[above=2] at (\xy{3}) {$g$};
		\node[below=2] at (\xy{4}) {$b$};
		\node[below=2] at (\xy{5}) {$c$};
		\node[below=2] at (\xy{6}) {$d$};
		\node[above=2] at (\xy{7}) {$e$};
		\end{tikzpicture}
		}&
		\parbox{.49\textwidth}{
		\centering
		\renewcommand{\vertexset}{(14,2,0),(15,2,1),(23,1,0),(24,4,0),(25,0,1),(27,3,1),(35,2.5,2),(36,3,0),(45,2,3),(46,0,0),(47,1.5,2),(56,4,1),(67,1,1)}
		\renewcommand{\edgeset}{(14,23),(14,36),(15,27),(15,67),(23,46),(23,67),(24,36),(25,67),(27,36),(27,56),(35,47)}
		\begin{tikzpicture}
		\drawgraph
		\node at (\xy{14}) {$ab$};
		\node at (\xy{15}) {$ac$};
		\node at (\xy{23}) {$fg$};
		\node at (\xy{24}) {$bf$};
		\node at (\xy{25}) {$cf$};
		\node at (\xy{27}) {$ef$};
		\node at (\xy{35}) {$cg$};
		\node at (\xy{36}) {$dg$};
		\node at (\xy{45}) {$bc$};
		\node at (\xy{46}) {$bd$};
		\node at (\xy{47}) {$be$};
		\node at (\xy{56}) {$cd$};
		\node at (\xy{67}) {$de$};
		\end{tikzpicture}
		}\vspace{.1in}\\
		\centering (a)&\centering (b)
		\end{tabular}
	\caption{(a) A graph $G$, and (b) the auxiliary graph $G^*$ of $G$.}\label{fig:auxgraph}
	\end{figure}
	Chv\'atal and Hammer observed that since for any subgraph $H$ of $G$ that is a threshold graph, $E(H)$ is an independent set in $G^*$, the following lower bound on $\thd(G)$ holds.

	\begin{lemma}[Chv\'atal-Hammer]\label{bipartite}
	$\thd(G)\geq\chi(G^*)$.
	\end{lemma}
	
	 This gave rise to the question of whether there is any graph $G$ such that $\thd(G)>\chi(G^*)$. Cozzens and Leibowitz~\cite{cozzens1984threshold} showed the existence of such graphs. In particular, they showed that for every $k\geq 4$, there exists a graph $G$ such that $\chi(G^*)=k$ but $\thd(G)>k$. The question of whether such graphs exist for $k=2$
	remained open for a decade (see~\cite{ma19942}). Ibaraki and Peled~\cite{ibarakipeled} showed, by means of some very involved proofs, that if $G$ is a split graph or if $G^*$ contains at most two non-trivial components, then $\chi(G^*)=2$ if and only if $\thd(G)=2$. They further conjectured that for any graph $G$, $\chi(G^*)=2\Leftrightarrow\thd(G)=2$.
	If the conjecture held, it would show immediately that graphs having a threshold cover of size 2 can be recognized in polynomial time, since the auxiliary graph $G^*$ can be constructed and its bipartiteness checked in polynomial time. In contrast, Yannakakis~\cite{yannakakis1982complexity} showed that it is NP-complete to recognize graphs having a threshold cover of size 3. Cozzens and Halsey~\cite{cozzenshalsey} studied some properties of graphs having a threshold cover of size 2 and showed that it can be decided in polynomial time whether the complement of a bipartite graph has a threshold cover of size 2.
	Finally, in 1995, Raschle and Simon~\cite{raschle1995recognition} proved the conjecture of Ibaraki and Peled by extending the methods in~\cite{ibarakipeled}.
	\begin{theorem}[Raschle-Simon]\label{thm:main-thm}
	For any graph $G$, $\chi(G^*)=2$ if and only if $\thd(G)=2$.
	\end{theorem}
	This proof of Raschle and Simon is very technical and involves the use of a number of complicated reductions and previously known results.
The aim of this paper is to provide a new, shorter and we believe, simpler, proof for Theorem~\ref{thm:main-thm}.

	Let $G$ be any graph such that $G^*$ is bipartite. We would like to prove that $G$ has a threshold cover of size 2. Note that if we take an arbitrary 2-coloring of $G^*$ having color classes $X_1,X_2$, the subgraph $G_i$ of $G$ formed by the edges in $X_i$ (where $i\in\{1,2\}$) need not necessarily be a threshold graph (this can be easily seen in the case when $G$ is a complete graph, as then $G^*$ contains only isolated vertices).

	We use a technique called the \emph{lexicographic method} introduced by Hell and Huang~\cite{hell1995lexicographic}, who demonstrated how this method can lead to shorter proofs and simpler recognition algorithms for certain problems that involve constructing a specific 2-coloring of an auxiliary bipartite graph whose vertices correspond to the edges of the graph.
	The method involves fixing an arbitrary ordering $<$ of the vertices of the graph, and then processing the edges in the ``lexicographic order'' implied by the ordering $<$.
	We adapt this technique to construct a 2-coloring of $G^*$ that can be used to generate a 2-threshold cover of $G$.
	However, it should be noted that unlike in the work of Hell and Huang, we start with a \emph{Lex-BFS ordering} of the vertices of the graph instead of an arbitrary ordering. It is an ordering of the vertices with the property that it is possible for a \emph{Lexicographic Breadth First Search (Lex-BFS)} algorithm to visit the vertices of the graph in that order. A Lex-BFS ordering is also a BFS ordering---i.e., a breadth-first search algorithm can also visit the vertices in that order---but it has some additional properties. Lex-BFS can be implemented to run in time linear in the size of the input graph and was introduced by Rose, Tarjan and Lueker~\cite{rose1976algorithmic} to construct a linear-time algorithm for recognizing chordal graphs. Later, Lex-BFS based algorithms were discovered for the recognition of many different graph classes (see~\cite{corneil2004lexicographic} for a survey).
	\section{Preliminaries}\label{prelims}
	Let $G=(V,E)$ be any graph. Recall that edges $ab,cd\in E(G)$ form an alternating 4-cycle if $bc,da\in E(\overline{G})$. In this case, we also say that $a,b,c,d,a$ is an alternating 4-cycle in $G$ (alternating 4-cycles are called $AC_4$s in~\cite{raschle1995recognition}). The edges $ab$ and $cd$ are said to be the \emph{opposite edges} of the alternating 4-cycle $a,b,c,d,a$. Thus for a graph $G$, the auxiliary graph $G^*$ is the graph with $V(G^*) = E(G)$ and $E(G^*)=\{\{ab,cd\}\colon ab,cd\in E(G)$ are the opposite edges of an alternating 4-cycle in $G\}$. Note that it follows from the definition of an alternating 4-cycle that if $a,b,c,d,a$ is an alternating 4-cycle, then the vertices $a,b,c,d$ are pairwise distinct. We shall refer to the vertex of $G^*$ corresponding to an edge $ab\in E(G)$ alternatively as $\{a,b\}$ or $ab$, depending upon the context. \smallskip
		
	Our goal is to provide a new proof for Theorem~\ref{thm:main-thm}.\smallskip
	
	It is easy to see that $\chi(G^*)=1$ if and only if $\thd(G)=1$. Therefore, by Lemma~\ref{bipartite}, it is enough to prove that if $G^*$ is bipartite, then $G$ can be covered by two threshold graphs. In order to prove this, we find a specific 2-coloring of the non-trivial components of $G^*$ (components of size at least 2) using the lexicographic method of Hell and Huang~\cite{hell1995lexicographic}.
	\medskip
	
	We say that $(A_0,A_1,A_2)$ is a \emph{valid 3-partition} of $E(G)$ if $\{A_0,A_1,A_2\}$ is a partition of $E(G)$ with the property that in any alternating 4-cycle in $G$, one of the opposite edges belongs to $A_1$ and the other to $A_2$. In other words, for any edge $\{ab,cd\}\in E(G^*)$, one of $ab,cd$ is in $A_1$ and the other in $A_2$.
				
	Given a valid 3-partition $(A_0,A_1,A_2)$ of $E(G)$ and $A\in\{A_1,A_2\}$, we say that $a,b,c,d$ is an \emph{alternating $A$-path} if $a\neq d$, $ab,cd\in A\cup A_0$, and $bc\in E(\overline{G})$. Further, we say that $a,b,c,d,e,f,a$ is an \emph{alternating $A$-circuit} if $a\neq d$, $ab,cd,ef\in A\cup A_0$, and $bc,de,fa\in E(\overline{G})$.
				
	\begin{observation}
	Let $(A_0,A_1,A_2)$ be a valid 3-partition of $E(G)$ and let $\{A,\oth{A}\}=\{A_1,A_2\}$.
	\begin{enumerate}
	\vspace{-0.06in}
	\renewcommand{\theenumi}{\alph{enumi}}
	\renewcommand{\labelenumi}{(\textit{\theenumi})}
	\item~\label{altpath} If $a,b,c,d$ is an alternating $A$-path, then $ad\in E(G)$.
	\item~\label{altcircuit} If $a,b,c,d,e,f,a$ is an alternating $A$-circuit, then $ef\in A$ and $ad\in\oth{A}$.
	\end{enumerate}
	\end{observation}
	\begin{proof}
	To prove~(\ref{altpath}), it just needs to be observed that if $ad\in E(\overline{G})$, then $a,b,c,d,a$ would be an alternating 4-cycle in $G$ whose opposite edges both belong to $A\cup A_0$, which contradicts the fact that $(A_0,A_1,A_2)$ is a valid 3-partition of $E(G)$. To prove~(\ref{altcircuit}), suppose that $a,b,c,d,e,f,a$ is an alternating $A$-circuit. Since $a,b,c,d$ is an alternating $A$-path, we have by~(\ref{altpath}) that $ad\in E(G)$. Then since $a,d,e,f,a$ is an alternating 4-cycle in $G$ and $ef\in A\cup A_0$, it follows that $ef\in A$ and $ad\in\oth{A}$.\hfill\qed
	\end{proof}
	
	We shall use the above observation throughout this paper without referring to it explicitly.
	\medskip
	
	Let $(A_0,A_1,A_2)$ be a valid 3-partition of $E(G)$ and let $\{A,\oth{A}\}=\{A_1,A_2\}$.
	We say that $(a,b,c,d,e)$ is an \emph{$A$-pentagon in $G$ with respect to $(A_0,A_1,A_2)$} if $a,b,c,d,e\in V(G)$, $ac,ad,be\in E(\overline{G})$, $ab,ae\in A$, $bc,bd,ec,ed\in\oth{A}$ and $cd\in A\cup A_0$. We abbreviate this to just ``$A$-pentagon'' when the graph $G$ and the 3-partition $(A_0,A_1,A_2)$ of $G$ are clear from the context. We say that an $A$-pentagon $(a,b,c,d,e)$ is a \emph{strict $A$-pentagon} if $cd\in A$. We say that $(a,b,c,d,e)$ is a \emph{pentagon} (resp. \emph{strict pentagon}) if it is an $A$-pentagon (resp. strict $A$-pentagon) for some $A\in\{A_1,A_2\}$. (Pentagons are similar to the ``$AP_5$-s'' in~\cite{raschle1995recognition}).
	
	 We say that $(x,y,z,w)$ is an \emph{$A$-switching path in $G$ with respect to $(A_0,A_1,A_2)$} if $x,y,z,w\in V(G)$, $xw\in E(\overline{G})$, $xy,zw\in A\cup A_0$, and $yz\in\oth{A}$. When the graph $G$ and the 3-partition $(A_0,A_1,A_2)$ of $G$ are clear from the context, we abbreviate this to just ``$A$-switching path''. We say that $(x,y,z,w)$ is a \emph{strict $A$-switching path} if it is an $A$-switching path and in addition, $xy,zw\in A$. We say that $(x,y,z,w)$ is a \emph{switching path} (resp. \emph{strict switching path}) if it is an $A$-switching path (resp. strict $A$-switching path) for some $A\in\{A_1,A_2\}$.
	 
	 Note that from the definitions of pentagons and switching paths, it follows that if $(a,b,c,d,e)$ is a pentagon, then the vertices $a,b,c,d,e$ are pairwise distinct, and if $(a,b,c,d)$ is a switching path, then the vertices $a,b,c,d$ are pairwise distinct.
	
	\begin{lemma} \label{lem:pathandpent}
		Let $(A_0,A_1,A_2)$ be a valid 3-partition of $E(G)$. Let $\{A,\oth{A}\}=\{A_1,A_2\}$. Let $(x,y,z,w)$ be an $A$-switching path in $G$ and let $y'z'\in E(G)$ be such that $yz',zy'\in E(\overline{G})$. Then,
		\begin{enumerate}
			\vspace{-0.06in}
			\renewcommand{\theenumi}{\alph{enumi}}
			\renewcommand{\labelenumi}{(\textit{\theenumi})}
			\item if $x=y'$, then $(x=y',y,z,w,z')$ is an $A$-pentagon and
			\item if $w=z'$, then $(w=z',z,y,x,y')$ is an $A$-pentagon.
		\end{enumerate}
	\end{lemma}
	\begin{proof}
		Since $yz\in \oth{A}$ and $\{yz,y'z'\}\in E(G^*)$ we have that $y'z'\in A$. Suppose that $x=y'$. Then $y,(x=y'),z,w,(x=y'),z',y$ is an alternating $A$-circuit (note that $y\neq w$ as $x\in N(y)\setminus N(w)$), implying that $yw\in \oth{A}$. This further implies that $z'\neq w$. Then we also have alternating $A$-circuits $z',y',z,w,x,y,z'$ and $z',(y'=x),w,z,(y'=x),y,z'$, implying that $xy\in A$ and $z'w,z'z\in \oth{A}$. Consequently, $(x=y',y,z,w,z')$ is an $A$-pentagon. Since $(w,z,y,x)$ is also an $A$-switching path, we can similarly conclude that if $w=z'$, then $(w=z',z,y,x,y')$ is an $A$-pentagon.\hfill\qed
	\end{proof}
\medskip

	Let $<$ be an ordering of the vertices of $G$.
	Given two $k$-element subsets $S=\{s_1,s_2,\ldots,s_k\}$ and $T=\{t_1,t_2,\ldots,t_k\}$ of $V(G)$, where $s_1< s_2<\cdots< s_k$ and $t_1< t_2<\cdots< t_k$, $S$ is said to be \emph{lexicographically smaller} than $T$, denoted by $S< T$,	if $s_j< t_j$ for some $j\in\{1,2,\ldots,k\}$, and $s_i=t_i$ for all $1\leq i < j$. In the usual way, we let $S\leq T$ denote the fact that either $S<T$ or $S=T$. For a set $S\subseteq V(G)$, we abbreviate $\min_< S$ to just $\min S$. Note that the relation $<$ (``is lexicographically smaller than'') that we have defined on $k$-element subsets of $V(G)$ is a total order. Therefore, given a collection of $k$-element subsets of $V(G)$, the lexicographically smallest one among them is well-defined.
\medskip
	
	The following observation states a well-known property of Lex-BFS orderings~\cite{corneil2004lexicographic}.
		
	\begin{observation}\label{obs:lexbfs}
	Let $<$ denote a Lex-BFS ordering of the vertices of a graph $G$. For $a,b,c\in V(G)$, if $a<b<c$, $ab\notin E(G)$ and $ac\in E(G)$, then there exists $x\in V(G)$ such that $x<a<b<c$, $xb\in E(G)$ and $xc\notin E(G)$.
	\end{observation}
	
	\section{Proof of Theorem~\ref{thm:main-thm}}\label{sec:proofofthm}
	Assume that $G^*$ is bipartite.
	
	We shall now construct a partial 2-coloring of the vertices of $G^*$ using the colors $\{1,2\}$ by means of an algorithm that consists of three phases. We shall describe the first two phases here, after which a partial 2-coloring of $G^*$ is obtained. The third phase, which will be described later, modifies this coloring so as to obtain a 2-threshold cover of $G$.
	\medskip
	
	\noindent\textbf{Phase I.} Construct a Lex-BFS ordering $<$ of $G$.
	\medskip
	
	\noindent Recall that every vertex of $G^*$ is a two-element subset of $V(G)$.
	
	\begin{tabbing}
	\noindent\textbf{Phase II.} \=For every non-trivial component $C$ of $G^*$, perform the following operation:\\[.05in]\>Choose the lexicographically smallest vertex in $C$ (with respect to the ordering $<$) and assign\\\>the color 1 to it. Extend this to a proper coloring of $C$ using the colors $\{1,2\}$.
	\end{tabbing}
	
	Note that after Phase II, every vertex of $G^*$ that is in a non-trivial component has been colored either 1 or 2. For $i\in\{1,2\}$, let $F_i=\{e\in V(G^*)\colon e$ is colored $i\}$. Further, let $F_0$ denote the set of all isolated vertices (trivial components) in $G^*$. Clearly, $F_0$ is exactly the set of uncolored vertices of $G^*$ and we have $V(G^*)=F_0\cup F_1\cup F_2$. Note that since the opposite edges of any alternating 4-cycle in $G$ correspond to adjacent vertices in $G^*$, one of them receives color 1 and the other color 2 in the partial 2-coloring of $G^*$ constructed in Phase II. It follows that $(F_0,F_1,F_2)$ is a valid 3-partition of $E(G)$.
		
	\subsection{No strict pentagons}\label{sec:nostrictpent}
	In this section we shall prove that there are no strict pentagons in $G$ with respect to $(F_0,F_1,F_2)$.

	Let $\{F,\oth{F}\}=\{F_1,F_2\}$. Let $(a,b,c,d,e)$ be a strict $F$-pentagon and $c_0d_0,c_1d_1,\ldots,c_kd_k$ be a path in $G^*$, where $c_0=c$, $d_0=d$, $k\geq 0$, and for each $i\in\{0,1,\ldots,k-1\}$, $c_id_{i+1},d_ic_{i+1}\in E(\overline{G})$. Since $cd=c_0d_0\in F$, it follows that $c_id_i\in F$ for all even $i$ and $c_id_i\in\oth{F}$ for all odd $i$.
	\begin{observation}\label{obs:cdtobe}
	For each $i\in\{0,1,\ldots,k\}$, the edges $c_ib,c_ie,d_ib,d_ie$ exist and they belong to $F$ when $i$ is odd and to $\oth{F}$ when $i$ is even.
	\end{observation}
	\begin{proof}
	We prove this by induction on $i$. This is easily seen to be true when $i=0$. Suppose that $i>0$. We shall assume without loss of generality that $i$ is odd as the other case is symmetric. Then by the induction hypothesis, $c_{i-1}b,d_{i-1}b,c_{i-1}e,d_{i-1}e\in\oth{F}$. Then $c_i,d_i,c_{i-1},b,e,d_{i-1},c_i$ is an alternating $\oth{F}$-circuit (note that $c_i\neq b$ as $d_{i-1}\in N(b)\setminus N(c_i)$), implying that $c_ib\in F$. By symmetric arguments, we get $c_ie,d_ib,d_ie\in F$.
	\hfill\qed
	\end{proof}

	\begin{remark}\label{rem:coincide}
	By the above observation, we have that:
	\begin{enumerate}
	\vspace{-0.06in}
	\renewcommand{\theenumi}{\alph{enumi}}
	\renewcommand{\labelenumi}{(\textit{\theenumi})}
	\item for each $i\in\{0,1,\ldots,k\}$, $c_i,d_i\notin\{b,e\}$,
	\item\label{it:parity} if $\{c_i,d_i\}\cap\{c_j,d_j\}\neq\emptyset$ for some $0\leq i,j\leq k$, then $i\equiv j\mod 2$, and
	\item\label{it:anotcd} for each even $i\in\{0,1,\ldots,k\}$, we have $a\notin\{c_i,d_i\}$.
	\end{enumerate} 
	\end{remark} 
	
	\begin{observation}\label{obs:apentagon}
	If $c_1\neq a$, then $(d,b,c_1,a,e)$ is a strict $\oth{F}$-pentagon. Similarly, if $d_1\neq a$, then $(c,b,d_1,a,e)$ is a strict $\oth{F}$-pentagon.
	\end{observation}
	\begin{proof}
	By Observation~\ref{obs:cdtobe}, we have $c_1b,c_1e,d_1b,d_1e\in F$.
	Suppose that $c_1\neq a$. Then $c_1,b,e,a,c,d,c_1$ is an alternating $F$-circuit, and therefore we have that $ac_1\in\oth{F}$. It now follows that $(d,b,c_1,a,e)$ is a strict $\oth{F}$-pentagon. By similar arguments, it can be seen that if $d_1\neq a$, then $ad_1\in\oth{F}$ and therefore $(c,b,d_1,a,e)$ is a strict $\oth{F}$-pentagon.
	\hfill\qed
	\end{proof}

	\begin{observation}\label{obs:cdprop}
	Let $S_0=\{a,c_0,d_0\}$ and for $1\leq i\leq k$, let $S_i=S_{i-1}\cup\{c_i,d_i\}$. Let $i\in\{0,1,\ldots,k\}$. For each $z\in\{c_i,d_i\}$, there exist $x_z,y_z\in S_i$ such that $(x_z,b,y_z,z,e)$ is a strict $F$-pentagon when $i$ is even and a strict $\oth{F}$-pentagon when $i$ is odd.
	\end{observation}
	\begin{proof}
	We are given an $i\in\{0,1,\ldots,k\}$ and a vertex $z$ that is either $c_i$ or $d_i$.
	First let us consider the case when $z=a$. Since $z\in\{c_i,d_i\}$, we have by Remark~\ref{rem:coincide}(\ref{it:anotcd}), that $i$ is odd, which implies that $i\geq 1$. Note that we have either $c_1\neq a$ or $d_1\neq a$. If $c_1\neq a$, we define $x_z=d$, $y_z=c_1$ and if $d_1\neq a$, we define $x_z=c$, $y_z=d_1$. Clearly, $x_z,y_z\in S_1\subseteq S_i$, since $i\geq 1$. By Observation~\ref{obs:apentagon}, we get that $(x_z,b,y_z,z,e)$ is a strict $\oth{F}$-pentagon, and so we are done. Therefore, we shall now assume that $z\neq a$.
	
	We shall now prove the statement of the lemma by induction on $i$.
	Clearly, when $i=0$, $z\in\{c_0,d_0\}$, so we can choose $x_z=a$, $y_z\in\{c,d\}\setminus\{z\}$ such that $(x_z,b,y_z,z,e)$ is a strict $F$-pentagon (note that $x_z,y_z\in S_0$ as required). So let us assume that $i\geq 1$.
	If $z\in\{c_j,d_j\}$ for some $j<i$, then by Remark~\ref{rem:coincide}(\ref{it:parity}) we have that $j\equiv i\mod 2$ and by the induction hypothesis applied to $j$ and $z$, there exist $x_z,y_z\in S_j\subseteq S_i$ (as $j<i$) such that $(x_z,b,y_z,z,e)$ is a strict $F$-pentagon if $i$ is even and a strict $\oth{F}$-pentagon if $i$ is odd, completing the proof. Therefore, we assume that there is no $j<i$ such that $z\in\{c_j,d_j\}$. Since we have already assumed that $z\neq a$, we now have $z\notin S_{i-1}$.

	Observe that there exists $z'\in\{c_{i-1},d_{i-1}\}$ such that $z'z\in E(\overline{G})$.
	Then by the induction hypothesis, there exist $x_{z'},y_{z'}\in S_{i-1}$ such that $(x_{z'},b,y_{z'},z',e)$ is a strict $F$-pentagon if $i-1$ is even and a strict $\oth{F}$-pentagon if $i-1$ is odd. 	
	Define $x_z=z'$ and $y_z=x_{z'}$. Then we have $x_z,y_z\in S_{i-1}\subseteq S_i$. Since $y_z\in S_{i-1}$ and $z\notin S_{i-1}$, we also have that $y_z\neq z$. Using Observation~\ref{obs:cdtobe} and the fact that $(x_{z'},b,y_{z'},z',e)$ is a strict $F$-pentagon (resp. $\oth{F}$-pentagon) if $i$ is odd (resp. even), we now have that $(y_z=x_{z'}),b,e,z,z',y_{z'},(x_{z'}=y_z)$ is an alternating $F$-circuit (resp. $\oth{F}$-circuit). Therefore, $y_zz\in\oth{F}$ if $i$ is odd and $y_zz\in F$ if $i$ is even. Consequently we get that $(x_z,b,y_z,z,e)$ is a strict $F$-pentagon when $i$ is even and a strict $\oth{F}$-pentagon when $i$ is odd.
	\hfill\qed
	\end{proof}
	
	It is easy to see that Observation~\ref{obs:cdprop} implies the following.
	
	\begin{remark}\label{rem:cdprop}
	Let $\{F,\oth{F}\}=\{F_1,F_2\}$ and let $(a,b,c,d,e)$ be any strict $F$-pentagon in $G$ with respect to $(F_0,F_1,F_2)$. Let $c'd'$ be a vertex in the same component as $cd$ in $G^*$. Then for each $z\in\{c',d'\}$, there exist $x_z,y_z\in V(G)$ such that $(x_z,b,y_z,z,e)$ is a strict $F$-pentagon if $c'd'\in F$ and a strict $\oth{F}$-pentagon if $c'd'\in\oth{F}$.
	\end{remark}

	Suppose that there is at least one strict pentagon in $G$ with respect to $(F_0,F_1,F_2)$.
	We say that a pentagon $(a,b,c,d,e)$ is lexicographically smaller than a pentagon $(a',b',c',d',e')$ if $\{a,b,c,d,e\}<\{a',b',c',d',e'\}$. Consider the lexicographically smallest strict pentagon $(a,b,c,d,e)$ in $G$. Let $\{F,\oth{F}\}=\{F_1,F_2\}$ such that $(a,b,c,d,e)$ is a strict $F$-pentagon.
	Since $cd\in F$, it belongs to a non-trivial component $C$ of $G^*$. Therefore, there exists $uv\in E(G)$ such that $cv,du\in E(\overline{G})$ (so that $\{cd,uv\}\in E(G^*)$). Clearly, at least one of $u,v$ is distinct from $a$. We assume without loss of generality that $u\neq a$ (by interchanging the labels of $c$ and $d$ if necessary). By applying Observation~\ref{obs:apentagon} to the path $(c_0d_0=cd),(c_1d_1=uv)$ in $G^*$, we get that $(d,b,u,a,e)$ is a strict $\oth{F}$-pentagon, which implies that $au\in\oth{F}$. By Observation~\ref{obs:cdtobe} applied to the same path, we get that $ub,ue\in F$.

	\begin{observation}\label{obs:amincd}
	$a>\min\{c,d\}$.
	\end{observation}
	\begin{proof}
	Suppose for the sake of contradiction that $a<\min\{c,d\}$. If $u<c$, then $(d,b,u,a,e)$ is a strict $\oth{F}$-pentagon that is lexicographically smaller than $(a,b,c,d,e)$, which is a contradiction. So we can assume that $c<u$, which gives us $a<c<u$. As $ac\in E(\overline{G})$ and $au\in E(G)$, by Observation~\ref{obs:lexbfs}, there exists a vertex $x$ such that $x<a<c<u$, $xc\in E(G)$ and $xu\in E(\overline{G})$. Since $a,u,x,c,a$ is an alternating 4-cycle in which $au\in\oth{F}$, we have that $xc\in F$. Then $b,a,c,x,u,e,b$ is an alternating $F$-circuit (note that $b\neq x$ as $u\in N(b)\setminus N(x)$), and therefore $xb\in\oth{F}$. Symmetrically, we also get that $xe\in\oth{F}$. Then $d,b,e,x,u,a,d$ is an alternating $\oth{F}$-circuit (note that $x\neq d$ as $x<a<\min\{c,d\}$), and therefore we have $xd\in F$. Now $(u,b,x,d,e)$ is a strict $F$-pentagon that is lexicographically smaller than $(a,b,c,d,e)$, which is a contradiction.
	\hfill\qed
	\end{proof}

	Let $c'd'$ be the lexicographically smallest vertex in $C$.

	\begin{observation}\label{obs:ckdkcd}
	$\min\{c',d'\}=\min\{c,d\}$.
    \end{observation}
    \begin{proof}
	We know that $c'd'\leq cd$, and therefore $\min\{c',d'\}\leq\min\{c,d\}$. Suppose that $z=\min\{c',d'\}<\min\{c,d\}$. From Remark~\ref{rem:cdprop}, we have that for each $z\in\{c',d'\}$, there exist vertices $x_z,y_z\in V(G)$ such that $(x_z,b,y_z,z,e)$ is a strict pentagon. Since $a>\min\{c,d\}$ by Observation~\ref{obs:amincd}, we have $a>z$. Then $(x_z,b,y_z,z,e)$ is a lexicographically smaller strict pentagon than $(a,b,c,d,e)$ which is a contradiction.\hfill\qed
	\end{proof}


    \begin{observation}\label{obs:cad}
    $a>\max\{c,d\}$.
    \end{observation}
    \begin{proof}
    Let $\{y,\oth{y}\}=\{c,d\}$ such that $y<\oth{y}$.
    By Observation~\ref{obs:amincd} it is now enough to show that $y<a<\oth{y}$ is not possible. Since $ya\in E(\overline{G})$ and $y\oth{y}\in E(G)$, $y<a<\oth{y}$ implies by Observation~\ref{obs:lexbfs} that there exists $x<y$ such that $xa\in E(G)$ but $x\oth{y}\in E(\overline{G})$. Then $x,a,y,\oth{y},x$ is an alternating 4-cycle, and therefore $xa$ and $y\oth{y}=cd$ belong to the same component $C$ of $G^*$. Thus $c'd'\leq xa$, which implies that $\min\{c',d'\}\leq\min\{x,a\}$. Since $\min\{x,a\}=x<y=\min\{c,d\}$, we now have $\min\{c',d'\}\leq\min\{x,a\}<\min\{c,d\}$. This contradicts Observation~\ref{obs:ckdkcd}.
    \hfill\qed
    \end{proof}
    
	Since $c'd'$ is the lexicographically smallest vertex in $C$, our algorithm would have colored it with the color~1. Therefore, we have $c'd'\in F_1$. Consider a path $c_0d_0,c_1d_1,\ldots,c_kd_k$ in $G^*$, where $c_0=c$, $d_0=d$, $c_k=c'$ and $d_k=d'$, in which for each $i\in\{0,1,\ldots,k-1\}$, $c_id_{i+1},d_ic_{i+1}\in E(\overline{G})$. Suppose that $cd\in F_2$. Then since $c_kd_k=c'd'\in F_1$, we have that $k$ is odd. Now by Remark~\ref{rem:coincide}(\ref{it:parity}), we have that $\{c_0,d_0\}\cap\{c_k,d_k\}=\emptyset$. But this contradicts Observation~\ref{obs:ckdkcd}. Thus we have that $cd\in F_1$. Therefore, $(a,b,c,d,e)$ is a strict $F_1$-pentagon, or in other words, $F=F_1$. Then, our earlier observations imply that $ub,ue\in F_1$ and $au\in F_2$.
	
	Since $ec,ab,ed$ and $bc,ae,bd$ are paths in $G^*$, it follows that $ec,ed$ lie in one component of $G^*$ and $bc,bd$ also lie in one component of $G^*$.
	Let $D$ be the component containing $bc,bd$ and $D'$ the component containing $ec,ed$ in $G^*$. Consider the lexicographically smallest vertex in $D\cup D'$. Let us assume without loss of generality that this vertex is in $D$ (we can interchange the labels of $b$ and $e$ if required). Define $p_0=b$, $q_0=c$. Then in $G^*$, there exists a path $p_0q_0,p_1q_1,\ldots,p_tq_t$ between $bc$ and the lexicographically smallest vertex $p_tq_t$ in $D$. As before, for $0\leq i\leq t-1$, we have $p_iq_{i+1},q_ip_{i+1}\in E(\overline{G})$ and for $0\leq i\leq t$, we have $p_iq_i\in F_1$ when $i$ is odd and $p_iq_i\in F_2$ when $i$ is even. Also, since $p_tq_t$ is the lexicographically smallest vertex in its component in $G^*$, we know that $p_tq_t\in F_1$, which implies that $t$ is odd.

	\begin{observation}\label{obs:bclemma}
	Let $i\in\{0,1,\ldots,t\}$. Then if $i$ is odd, we have
	\begin{enumerate}
	\vspace{-0.06in}
	\renewcommand{\theenumi}{\alph{enumi}}
	\renewcommand{\labelenumi}{(\textit{\theenumi})}
	\item $p_i\notin\{b,e\}$,
	\item $q_i\notin\{a,c,d\}$,
	\item $p_ib,p_ie\in F_1$,
	\item Either $p_i=a$ or $p_ia\in F_2$, and
	\item Either $q_ic\in F_2$ or $q_id\in F_2$.
	\end{enumerate}
	and if $i$ is even, we have
	\begin{enumerate}
	\vspace{-0.06in}
	\renewcommand{\theenumi}{\alph{enumi}}
	\renewcommand{\labelenumi}{(\textit{\theenumi})}
	\item $q_i\notin\{b,e\}$,
	\item $p_i\notin\{a,c,d\}$,
	\item $q_ib,q_ie\in F_2$,
	\item Either $q_i=d$ or $q_id\in F_1$, and
	\item Either $p_iu\in F_1$ or $p_ia\in F_1$.
	\end{enumerate}
	\end{observation}
\begin{proof}
	We shall prove this by induction on $i$. If $i=0$, then the statement of the lemma can be easily seen to be true. Suppose that $i>0$. We give a proof for the case when $i$ is odd (the case when $i$ is even is symmetric and can be proved using similar arguments). By the induction hypothesis,
	$q_{i-1}b,q_{i-1}e\in F_2$, and therefore since $p_iq_{i-1}\in E(\overline{G})$, we have $p_i\notin \{b,e\}$. We now prove the following claim.
	\begin{myclaim}\label{clm:equiv}
		For $x\in \{a,u\}$, if $p_i=x $ or $p_ix\in F_2$, then $p_ib,p_ie\in F_1$. 
	\end{myclaim}
	If $p_i=x$ then there is nothing to prove as we already know that $ab,ae,ub,ue\in F_1$. So assume that $p_ix\in F_2$. Let $\{z,\bar{z}\}=\{b,e\}$. Then $p_i,x,d,z,\bar{z},q_{i-1},p_i$ is an alternating $F_2$-circuit (recall that $p_i\notin\{b,e\}$), which implies that $p_iz\in F_1$. We thus get that $p_ib,p_ie\in F_1$. This proves the claim.
	
	\medskip
	By the induction hypothesis we know that either $p_{i-1}a\in F_1$ or $p_{i-1}u\in F_1$, and also that $p_{i-1}\notin\{c,d\}$. First suppose that $p_{i-1}a\in F_1$. This implies that $q_i\neq a$. Let $\{y,\bar{y}\}=\{c,d\}$. Then we have that $p_{i-1},a,\bar{y},y$ is an alternating $F_1$-path implying that $p_{i-1}y\in E(G)$. Thus, $p_{i-1}c,p_{i-1}d\in E(G)$. This implies that $q_i\notin \{c,d\}$. By the induction hypothesis we also have that $q_{i-1}y\in F_1$ for some $y\in\{c,d\}$. Then $q_i,p_i,q_{i-1},y,a,p_{i-1},q_i$ is an alternating $F_1$-circuit, which implies that $q_iy\in F_2$. If $p_i\neq a$, then $p_i,q_i,p_{i-1},a,y,q_{i-1},p_i$ is an alternating $F_1$-circuit, implying that $p_ia\in F_2$. Since we have either $p_i=a$ or $p_ia\in F_2$ we are done by Claim~\ref{clm:equiv}. 
	
	Therefore we can assume that $p_{i-1}a\notin F_1$. If $i=1$, then we know that $p_{i-1}a=ba\in F_1$, so we can assume that $i\geq 2$. By the induction hypothesis, we have that for some $y\in\{c,d\}$, $q_{i-2}y\in F_2$. Therefore if $p_{i-1}a\in E(G)$, then we have that $p_{i-1},a,y,q_{i-2},p_{i-1}$ is an alternating 4-cycle in which $q_{i-2}y\in F_2$, implying that $p_{i-1}a\in F_1$ which is a contradiction. Since we know that $p_{i-1}\neq a$ by the induction hypothesis, we can assume that $p_{i-1}a\in E(\overline{G})$. Note that since $p_{i-1}a\notin F_1$, we have by the induction hypothesis that $p_{i-1}u \in F_1$. If $q_{i-1}=d$, then $p_{i-1},(q_{i-1}=d),u,a,p_{i-1}$ is an alternating 4-cycle whose opposite edges both belong to $F_2$, which is a contradiction. Therefore by the induction hypothesis we have $q_{i-1}d \in F_1$. If $q_i=a$ (resp. $q_i=c$) then $p_i,(q_i=a),d,q_{i-1},p_i$ (resp. $p_{i-1},u,d,(c=q_i),p_{i-1}$) is an alternating 4-cycle whose opposite edges are both in $F_1$, which is a contradiction. Therefore, $q_i\notin\{a,c\}$. If $p_ia\in F_2$ then we have that $a,p_i,q_{i-1},p_{i-1},a$ is an alternating 4-cycle whose opposite edges are both in $F_2$, which is a contradiction. This implies that $p_ia\notin F_2$ and therefore $p_i\neq u$. Then $p_i,q_i,p_{i-1},u,d,q_{i-1},p_i$ is an alternating $F_1$-circuit, implying that $p_iu\in F_2$. Therefore by Claim~\ref{clm:equiv}, we have that $p_ib,p_ie\in F_1$. Now if $a\neq p_i$, then $p_i,b,e,a,d,q_{i-1},p_i$ is an alternating $F_1$-circuit, which implies that $p_ia\in F_2$ which is a contradiction. This implies that $a= p_i$, which further implies that $q_i\neq d$. Then $q_i,p_i,q_{i-1},d,u,p_{i-1},q_i$ is an alternating $F_1$-circuit, which implies that $q_id\in F_2$ and we are done. 
	\hfill\qed
\end{proof}

    \begin{observation}\label{obs:apdq}
	For each even $i\in \{0,1,2,\ldots,t\}$, either $ap_i\in E(G)$ or both $dq_{i-1},dq_{i+1}\in E(G)$.
	\end{observation}
	\begin{proof}
    Suppose that there exists an even $i\in \{0,1,2,\ldots,t\}$ and $j\in\{i-1,i+1\}$ such that $ap_i,dq_j\notin E(G)$. By Observation~\ref{obs:bclemma}, we know that $p_i\neq a$ and $q_j\neq d$. So we have $ap_i,dq_j\in E(\overline{G})$. Now if $d\neq q_i$, then we have by Observation~\ref{obs:bclemma} that $q_id\in F_1$. Then $p_j,q_j,d,q_i,p_j$ is an alternating 4-cycle whose both opposite edges belong to $F_1$, which is a contradiction. Therefore we can assume that $d=q_i$. Then $(d=q_i),p_i,a,u,(d=q_i)$ is an alternating 4-cycle whose opposite edges both belong to $F_2$, which is again a contradiction.\hfill\qed
	\end{proof}
	
	Recall that $D'$ is the component containing $ec$ in $G^*$.
	
	\begin{observation}\label{obs:pathab}
	For any odd $i\in\{0,1,\ldots,t\}$, if $ap_{i-1}\in E(G)$, then for each $y\in\{c,d\}$ for which $yq_i\in E(G)$, we have $yq_i\in D'$. On the other hand, if $ap_{i-1}\notin E(G)$, then $dq_i\in D'$.
	\end{observation}
	\begin{proof}
	We prove this by induction on $i$. When $i=1$, we have $ap_0=ab\in E(G)$ and for each $y\in\{c,d\}$ such that $yq_1\in E(G)$, we have that $ec,(ab=ap_0),yq_1$ is a path in $G^*$. We thus have the base case.
	We shall now prove the claim for $i\geq 3$ assuming that the claim is true for $i-2$.
	Suppose that $ap_{i-1}\in E(G)$. By Observation~\ref{obs:bclemma}, there exists $y''\in\{c,d\}$ such that $y''q_{i-2}\in E(G)$. By the induction hypothesis, either $y''q_{i-2}\in D'$ or $dq_{i-2}\in D'$ (depending upon whether $ap_{i-3}$ is an edge or not). Thus in any case, we have that there exists $y'\in\{c,d\}$ such that $y'q_{i-2}\in D'$. Now for each $y\in\{c,d\}$ such that $yq_i\in E(G)$, since $y'q_{i-2},ap_{i-1},yq_i$ is a path in $G^*$, we get that $yq_i\in D'$, so we are done. Next, suppose that $ap_{i-1}\notin E(G)$. Then by Observation~\ref{obs:bclemma}, we have $up_{i-1}\in E(G)$ and by Observation~\ref{obs:apdq}, we have $dq_{i-2},dq_i\in E(G)$. We then have by the induction hypothesis that $dq_{i-2}\in D'$. Since $dq_{i-2},up_{i-1},dq_i$ is a path in $G^*$, we have $dq_i\in D'$.\hfill\qed
	\end{proof}
   
   Recall that $C$ is the component of $G^*$ containing the vertex $cd$.
\begin{observation} \label{obs:apjcd}
For each odd $i\in\{0,1,\ldots,t\}$, if $a\neq p_i$ then $ap_i\in C$.
\end{observation}
\begin{proof}
We prove this by induction on $i$. The base case when $i=1$ is true since if $a\neq p_1$ then by Observation~\ref{obs:bclemma}, $ap_1\in E(G)$, and since $\{ap_1,(q_0=c)d\}\in E(G^*)$, we have $ap_1\in C$. Assume that $i\geq 3$ and the claim is true for $i-2$. Suppose that $a\neq p_i$. Then we have $ap_i\in E(G)$ by Observation~\ref{obs:bclemma}. If $d=q_{i-1}$ then we have $\{ap_i,c(q_{i-1}=d)\}\in E(G^*)$, so we have $ap_i\in C$. So we assume that $d\neq q_{i-1}$. Then by Observation~\ref{obs:bclemma}, we have that $dq_{i-1}\in E(G)$. By the induction hypothesis, we have that either $ap_{i-2}\in C$ or $a=p_{i-2}$. If $ap_{i-2}\in C$, then since $ap_i,dq_{i-1},ap_{i-2}$ is a path in $G^*$, we have $ap_i\in C$. On the other hand, if $a=p_{i-2}$ then we again have $ap_i\in C$ as $ap_i,dq_{i-1},u(p_{i-2}=a),cd$ is a path in $G^*$.
\hfill\qed
\end{proof}

Recall that $t$ is odd, $p_tq_t\in D$, and $p_tq_t$ is the lexicographically smallest vertex in $D\cup D'$.

    \begin{observation}\label{obs:ptcd}
    $p_t< \min\{c,d\}$ 
    \end{observation}
    \begin{proof}
    Let $\{y,\bar{y}\}=\{c,d\}$, where $y<\bar{y}$. Note that $p_t\notin \{c,d\}$, since by Observation~\ref{obs:bclemma}, $p_tb\in F_1$, but we know that $cb,db\in F_2$. By the same lemma, we also have that $q_t\notin \{c,d\}$. Therefore as $\min\{p_t,q_t\}\leq \min\{c,d\}$ (since $p_tq_t<bc,bd$), we have that $\min\{p_t,q_t\}< \min\{c,d\}=y$. Now if $p_t=\min\{p_t,q_t\}$ then we are done. Therefore let us assume that $q_t=\min\{p_t,q_t\}$, and so $q_t<y$.
    
    Suppose that $yq_t\in E(G)$. If $yq_t\notin D'$, then by Observation~\ref{obs:pathab}, we have that $ap_{t-1}\notin E(G)$ and $\oth{y}q_t\in D'$. By Observation~\ref{obs:bclemma}, we know that $p_{t-1}\neq a$, which implies that $ap_{t-1}\in E(\overline{G})$. By our choice of $p_tq_t$, we now have that $p_tq_t<\oth{y}q_t$, which implies that $p_t<\oth{y}$. Now by Observation~\ref{obs:cad}, $p_t\neq a$, which implies by Observation~\ref{obs:bclemma} that $p_ta\in F_2$. Then $a,p_t,q_{t-1},p_{t-1},a$ is an alternating 4-cycle in which both opposite edges belong to $F_2$, which is a contradiction. We can thus conclude that $yq_t\in D'$. Then by our choice of $p_tq_t$, we have that $p_t<y$, and we are done. So we assume that $yq_t\notin E(G)$. 
        
    Recall that $q_t<y$ (and therefore $yq_t\in E(\overline{G})$). Now if $y<p_t$ then we have $q_t<y<p_t$ where $q_ty\notin E(G)$ and $q_tp_t\in E(G)$. By Observation~\ref{obs:lexbfs}, this implies that there exists $x<q_t$ such that $xy\in E(G)$ and $xp_t\notin E(G)$ (which means that $xp_t\in E(\overline{G})$ since $x<p_t$). Then $\{xy,p_tq_t\}\in E(G^*)$, which implies that $xy\in D$. But $xy<p_tq_t$, which contradicts our choice of $p_tq_t$. We can thus conclude that $p_t<y$ (recall that $p_t\neq y$ as $p_t\notin\{c,d\}$) and we are done.
    \hfill\qed
   \end{proof}

Note that by Observation~\ref{obs:ptcd} and Observation~\ref{obs:cad} we have that $a\neq p_t$. Then by Observation~\ref{obs:apjcd}, we have $ap_t\in C$. By Observation~\ref{obs:ptcd} and Observation~\ref{obs:ckdkcd}, $p_t<\min\{c',d'\}$, which implies that $ap_t<c'd'$. This is a contradiction to our choice of $c'd'$. Therefore we have the following lemma.
 
\begin{lemma}\label{lem:nostrictpent}
There are no strict pentagons in $G$ (with respect to $(F_0,F_1,F_2)$).
\end{lemma}

\subsection{No strict switching paths}
In this section, we show that there are no strict switching paths either in $G$ with respect to $(F_0,F_1,F_2)$.
First we note the following observation.
 
 \begin{observation} \label{obs:pathandpent}
 Let $(x,y,z,w)$ be a strict switching path with respect to $(F_0,F_1,F_2)$. Let $y'z'\in E(G)$ be such that $yz',zy'\in E(\overline{G})$. Then, $y'\neq x$ and $z'\neq w$.
 \end{observation}
 \begin{proof}
 	 Let $\{F,\oth{F}\}=\{F_1,F_2\}$. Suppose that $(x,y,z,w)$ is a strict $F$-switching path. Then we have that $xy,zw\in F$, $yz\in \oth{F}$, and $xw\in E(\overline{G})$. By Lemma~\ref{lem:pathandpent} and the fact that $zw,xy\in F$, we know that if $y'=x$ then $(x=y',y,z,w,z')$ is a strict $F$-pentagon, and if $z'=w$ then $(w=z',z,y,x,y')$ is a strict $F$-pentagon. Since we know by Lemma~\ref{lem:nostrictpent} that there are no strict pentagons in $G$, we can conclude that $y'\neq x$ and $z'\neq w$.\hfill\qed
 \end{proof}
 
 Now we show that there are no strict switching paths in $G$. Suppose not.
 Then let $(a,b,c,d)$ be the lexicographically smallest strict switching path in $G$.

 \begin{observation}\label{type1}
 $(a,b,c,d)$ is not a strict $F_1$-switching path.
 \end{observation}
 \begin{proof}
 	Suppose for the sake of contradiction that $(a,b,c,d)$ is a strict $F_1$-switching path. Let $C$ be the component of $G^*$ containing $bc$. Let $b_0c_0,b_1c_1,\ldots,b_kc_k$, where $b_0=b$ and $c_0=c$, be a path in $C$ between $bc$ and the lexicographically smallest vertex $b_kc_k$ in $C$. We assume that for each $i\in\{0,1,\ldots,k-1\}$, $b_ic_{i+1},c_ib_{i+1}\in E(\overline{G})$. As $b_0c_0\in F_2$, it follows that $b_ic_i\in F_2$ for each even $i$ and $b_ic_i\in F_1$ for each odd $i$. Since $b_kc_k$ is the lexicographically smallest vertex in its component in $G^*$, we know that $b_kc_k\in F_1$, which implies that $k$ is odd.
 	
 	We claim that $b_ia,c_id\in F_1$ for each even $i$ and $b_ia,c_id\in F_2$ for each odd $i$, where $0\leq i\leq k$. We prove this by induction on $i$. The case where $i=0$ is trivial as $b_0=b$ and $c_0=c$. So let us assume that $i>0$. Consider the case where $i$ is odd. As $i-1$ is even, by the induction hypothesis we have $b_{i-1}a,c_{i-1}d\in F_1$. Since $b_{i-1}c_{i-1}\in F_2$, we can observe that, $(a,b_{i-1},c_{i-1},d)$ is a strict $F_1$-switching path. Then by Observation~\ref{obs:pathandpent}, we have that $a\neq b_i$ and $d\neq c_i$. Now the alternating $F_1$-circuits $b_i,c_i,b_{i-1},a,d,c_{i-1},b_i$ and $c_i,b_i,c_{i-1},d,a,b_{i-1},c_i$ imply that $b_ia,c_id \in F_2$. The case where $i$ is even is symmetric and hence the claim. 
 	
 	By the above claim, $b_ka,c_kd\in F_2$. Since $b_kc_k\in F_1$, we now have that $(a,b_k,c_k,d)$ is a strict $F_2$-switching path. Since $b_kc_k<bc$, we have that $\{a,b_k,c_k,d\}<\{a,b,c,d\}$, which is a contradiction to our assumption that $(a,b,c,d)$ is the lexicographically smallest strict switching path in $G$.\qed
 \end{proof}
 \begin{observation}\label{type2}
 $(a,b,c,d)$ is not a strict $F_2$-switching path.
 \end{observation}
 \begin{proof}
 	Suppose for the sake of contradiction that $(a,b,c,d)$ is a strict $F_2$-switching path. By the symmetry between $a$ and $d$, we can assume without loss of generality that $a<d$. 
 	
 	As $bc\in F_1$, the vertex $bc$ belongs to a non-trivial component of $G^*$. Then there exists a neighbor $uv$ of $bc$ in $G^*$ such that $bv,uc\in E(\overline{G})$. As $bc\in F_1$, we have $uv\in F_2$. By Observation~\ref{obs:pathandpent}, we have that $u\neq a$. Then $a,b,v,u,c,d,a$ is an alternating $F_2$-circuit, implying that $au\in F_1$. As $ab\in F_2$, we know that $ab$ is not the lexicographically smallest vertex in its component. Let $a_0b_0,a_1b_1,\ldots,a_kb_k$ be a path in $G^*$ between $ab$ and the lexicographically smallest vertex $a_kb_k$ in its component, where $a_0=a$, $b_0=b$, and for $0\leq i<k$, $a_ib_{i+1},a_{i+1}b_i\in E(\overline{G})$. Note that for $0\leq i\leq k$, $a_ib_i\in F_2$ if $i$ is even and $a_ib_i\in F_1$ if $i$ is odd. Since $a_kb_k\in F_1$ (as it is the lexicographically smallest vertex in its component in $G^*$), this implies that $k$ is odd.
 	
	We claim that for $0\leq i\leq k$, $a_iu,b_ic\in F_1$ if $i$ is even and $a_iu,b_ic\in F_2$ if $i$ is odd. We prove this by induction on $i$.
 	The base case when $i=0$ is trivial, since $au,bc\in F_1$. Let $i>0$ be odd. By the induction hypothesis we have that $a_{i-1}u,b_{i-1}c\in F_1$. Since $a_{i-1}b_{i-1}\in F_2$ we can observe that $(u,a_{i-1},b_{i-1},c)$ is a strict $F_1$-switching path. Therefore by Observation~\ref{obs:pathandpent}, we have that $a_i\neq u$ and $b_i\neq c$. Then we have alternating $F_1$-circuits $a_i,b_i,a_{i-1},u,c,b_{i-1},a_i$ and $b_i,a_i,b_{i-1},c,u,a_{i-1},b_i$, implying that $a_iu,b_ic\in F_2$. The case when $i$ is even is symmetric. This proves our claim. Since $k$ is odd, we now have that $a_ku,b_kc \in F_2$. Note that now $(c,b_k,a_k,u)$ is a strict $F_2$-switching path.
 	
 	Suppose that $d<b$. Then we have that $a<d<b$, where $ad\in E(\overline{G})$ and $ab\in E(G)$. Therefore by Observation~\ref{obs:lexbfs}, there exists $x<a$ such that $xd\in E(G)$ and $xb\in E(\overline{G})$. Then $x,d,a,b,x$ is an alternating 4-cycle in which $ab\in F_2$, implying that $xd\in F_1$. Then we have a strict $F_1$-switching path $(x,d,c,b)$ such that $\{x,d,c,b\}<\{a,b,c,d\}$, which is a contradiction to the choice of $(a,b,c,d)$. Therefore we can assume that $b<d$.
 	Since $a_kb_k< ab$ and $a,b<d$, we have that $\{c,b_k,a_k,u\}<\{a,b,c,d\}$. As $(c,b_k,a_k,u)$ is a strict switching path, this contradicts the choice of $(a,b,c,d)$.
 	\hfill\qed
 \end{proof}
 \medskip
 
 From Observation~\ref{type1} and Observation~\ref{type2}, we have the following lemma.
 \begin{lemma}\label{lem:nostrictpaths}
 There are no strict switching paths in $G$ (with respect to $(F_0,F_1,F_2)$).
 \end{lemma}

\noindent\textbf{Remark.} Given a 2-coloring of $G^*$ in which the color classes are denoted by $E_1$ and $E_2$, Raschle and Simon~\cite{raschle1995recognition} define an ``$AP_6$'' in $G$ to be a sequence $v_0,v_1,\ldots,v_5,v_0$ of distinct vertices of $G$ such that $v_0v_1,v_2v_3,v_4v_5\in E_i$ for some $i\in\{1,2\}$ and $v_1v_2,v_3v_4,v_5v_0\in E(\overline{G})$. A 2-coloring of $G^*$ is said to be ``$AP_6$-free'' if there is no $AP_6$ in $G$ with respect to that coloring.
Raschle and Simon observe that if $G^*$ has an $AP_6$-free 2-coloring, then $G$ has a 2-threshold cover and it can be computed in time $O(|E(G)|^2)$ (using Theorem~3.1, Theorem~2.5, Fact~2 and Fact~1 in~\cite{raschle1995recognition}). The major part of the work of Raschle and Simon is to show that an $AP_6$-free 2-coloring of $G^*$ always exists if $G^*$ is bipartite and that it can be computed in time $O(|E(G)|^2)$ (Sections~3.2 and~3.3 of~\cite{raschle1995recognition}). 
It can be seen that any 2-coloring of $G^*$ obtained by extending the partial 2-coloring of $G^*$ computed after Phases~I and~II of our algorithm is in fact an $AP_6$-free 2-coloring of $G^*$ as follows. Let $E_1$ and $E_2$ be the color classes of such a 2-coloring of $G^*$. We can assume without loss of generality that $F_1\subseteq E_1$ and $F_2\subseteq E_2$. Note that $F_0\subseteq E_1\cup E_2$. Suppose that there is an $AP_6$ $v_0,v_1,\ldots,v_5,v_0$ in $G$ with respect to this coloring where the edges $v_0v_1,v_2v_3,v_4v_5\in E_i$, where $i\in\{1,2\}$. Note that $(\emptyset,E_1,E_2)$ is a valid 3-partition of $E(G)$. For each even $j\in\{0,1,\ldots,5\}$, since $v_j,v_{j+1},v_{j+2},v_{j+3}$ (subscripts modulo 6) is an alternating $E_i$-path, we have that $v_jv_{j+3}\in E(G)$. This implies that for each even $j\in\{0,1,\ldots,5\}$, $v_j,v_{j+1},v_{j+2},(v_{j+5}=v_{j-1}),v_j$ is an alternating 4-cycle in $G$ (note that from the previous observation, we have $v_{j+2}v_{j+5}\in E(G)$), from which it follows that $v_jv_{j+1}$ is in a non-trivial component of $G^*$. Therefore, $v_0v_1,v_2v_3,v_4v_5\notin F_0$. Since these edges belong to $E_i$, it follows that $v_0v_1,v_2v_3,v_4v_5\in F_i$. Then $v_0,v_1,\ldots,v_5,v_0$ is an alternating $F_i$-circuit, and therefore $v_0v_3\in F_{3-i}$. This implies that $(v_2,v_3,v_0,v_1)$ is a strict $F_i$-switching path in $G$, which contradicts Lemma~\ref{lem:nostrictpaths}. Thus the proof of Theorem~\ref{thm:main-thm} can already be completed using the observations in~\cite{raschle1995recognition}. In the next section, we nevertheless give a self-contained proof that shows that $G$ has a 2-threshold cover without using the ``threshold completion'' method used in~\cite{ibarakipeled,raschle1995recognition}. Also note that since it is clear that Phases~I and~II of the algorithm, and also the initial construction of $G^*$, can be done in time $O(|E(G)|^2)$, we have an algorithm with the same time complexity that computes the 2-threshold cover of a graph $G$ whose auxiliary graph $G^*$ is bipartite (note however that there is a faster algorithm for computing a 2-threshold cover due to Sterbini and Raschle~\cite{sterbini1998n3}).

	\subsection{Constructing the 2-threshold cover of $G$}
	
    \begin{observation}\label{obs:nodoublepent}
	There does not exist $a_1,a_2,b_1,b_2,e_1,e_2,c,d\in V(G)$ such that $(a_1,b_1,c,d,e_1)$ is an $F_1$-pentagon and $(a_2,b_2,c,d,e_2)$ is an $F_2$-pentagon.
    \end{observation}
    \begin{proof}
    Suppose not. Then as $b_1c,e_1c\in F_2$ and $b_2c,e_2c\in F_1$, we have $\{b_1,e_1\}\cap \{b_2,e_2\}=\emptyset$.
    Then $b_1,a_1,c,b_2$ and $e_1,a_1,c,e_2$ are alternating $F_1$-paths, implying that $b_1b_2,e_1e_2\in E(G)$. As $b_1,b_2,e_2,e_1,b_1$ is an alternating 4-cycle, we have $\{b_1b_2,e_1e_2\}\in E(G^*)$. Thus, $b_1b_2\notin F_0$, or in other words, $b_1b_2\in F_1\cup F_2$. If $b_1b_2\in F_1$, then $(c,b_1,b_2,a_2)$ is a strict $F_2$-switching path, which contradicts Lemma~\ref{lem:nostrictpaths}. On the other hand, if $b_1b_2\in F_2$, then $(c,b_2,b_1,a_1)$ is a strict $F_1$-switching path, which again gives a contradiction to Lemma~\ref{lem:nostrictpaths}.
    \hfill\qed
	\end{proof}

	We shall now describe Phase~III of the algorithm that yields a partial 2-coloring of $G^*$ that can be directly converted into a 2-threshold cover of $G$.
	\begin{tabbing}
	\textbf{Phase~III.} \=For \=each $i\in\{1,2\}$, let\\[.05in]\>\>$S_i=\{cd\in F_0\colon \exists a,b,e\in V(G)$ such that $(a,b,c,d,e)$ is an $F_i$-pentagon in $G$ with\\\>\>respect to $(F_0,F_1,F_2)\}$.\\[0.05in]\>Color every vertex in $S_1$ with 2 and every vertex in $S_2$ with 1.
	\end{tabbing}
	
	Let $F'_0$ be the set of vertices of $G^*$ that are uncolored after Phase III, and for $i\in\{1,2\}$, let $F'_i$ be the set of vertices of $G^*$ that are colored $i$. Clearly, $F'_0=F_0\setminus (S_1\cup S_2)$, $F'_1=F_1\cup S_2$ and $F'_2=F_2\cup S_1$.	Note that $S_1,S_2\subseteq F_0$ and that $S_1\cap S_2=\emptyset$ by Observation~\ref{obs:nodoublepent}. It is easy to see that $\{F'_0,F'_1,F'_2\}$ is a partition of $E(G)$. Further, since $F_1\subseteq F'_1$, $F_2\subseteq F'_2$ and $(F_0,F_1,F_2)$ is a valid 3-partition of $E(G)$, it follows that $(F'_0,F'_1,F'_2)$ is also a valid 3-partition of $E(G)$. We shall show that $(V(G),F'_0\cup F'_1)$ and $(V(G),F'_0\cup F'_2)$ are both threshold graphs, thereby completing the proof of Theorem~\ref{thm:main-thm}. From here onwards, we use the terms ``pentagons'' and ``switching paths'' with respect to $(F'_0,F'_1,F'_2)$ unless otherwise mentioned.
	
	\begin{lemma}\label{lem:nopentagons}
	There are no pentagons in $G$ with respect to $(F'_0,F'_1,F'_2)$.
	\end{lemma}
	\begin{proof}
	Suppose for the sake of contradiction that $(a,b,c,d,e)$ is a pentagon in $G$ with respect to $(F'_0,F'_1,F'_2)$. Let $i\in\{1,2\}$ such that $(a,b,c,d,e)$ is an $F'_i$-pentagon. Recall that $ec,ab,ed$ and $bc,ae,bd$ are paths in $G^*$ and hence each of $ab,ae,bc,bd,ec,ed$ is in a non-trivial component of $G^*$. Thus none of them is in $F_0$. Since $ab,ae\in F'_i$ and $bc,bd,ec,ed\in F'_{3-i}$, this implies that $ab,ae\in F_i$ and $bc,bd,ec,ed\in F_{3-i}$.
	Since $(a,b,c,d,e)$ is an $F'_i$-pentagon, we have $cd\in F'_0\cup F'_i$. This implies that $cd\notin F'_{3-i}$ and that $cd\in F_0\cup F_i$. If $cd\in F_0$, then $(a,b,c,d,e)$ is an $F_i$-pentagon in $G$ with respect to $(F_0,F_1,F_2)$, which implies that $cd\in S_i$, and therefore $cd\in F'_{3-i}$. Since this is a contradiction, we can assume that $cd\in F_i$. Then $(a,b,c,d,e)$ is a strict $F_i$-pentagon in $G$ with respect to $(F_0,F_1,F_2)$, contradicting Lemma~\ref{lem:nostrictpent}.
	\hfill\qed
	\end{proof}

	\begin{lemma}\label{lem:nopath}
    There are no switching paths in $G$ with respect to $(F'_0,F'_1,F'_2)$.
	\end{lemma}
    \begin{proof}
    Suppose not. Let $(a,b,c,d)$ be a switching path in $G$ with respect to $(F'_0,F'_1,F'_2)$. Let $i\in\{1,2\}$ such that $(a,b,c,d)$ is an $F'_i$-switching path. Then we have $ad\in E(\overline{G})$, $ab,cd\in F'_i\cup F'_0$, and $bc\in F'_{3-i}$. Suppose that $bc$ belongs to a non-trivial component of $G^*$. Then there exists $uv\in E(G)$ such that $bv,cu\in E(\overline{G})$. By Lemma~\ref{lem:pathandpent} and Lemma~\ref{lem:nopentagons}, we have that $a\neq u$ and $d\neq v$. Notice that since $bc\in F'_{3-i}$ and $b,c,u,v,b$ is an alternating 4-cycle, we have $uv\in F'_i$. Then $d,c,u,v,b,a,d$ and $a,b,v,u,c,d,a$ are alternating $F'_i$-circuits, implying that $dv,au\in  F'_{3-i}$ and $ab,cd\in F'_i$. This further implies that $(a,b,c,d)$ is a strict $F'_i$-switching path with respect to $(F'_0,F'_1,F'_2)$. Since $b,a,d,v,b$ and $c,d,a,u,c$ and $b,c,u,v,b$ are alternating 4-cycles, we also have that $ab,cd,bc\notin F_0$, which further implies that $ab,cd\in F_i$ and $bc\in F_{3-i}$. Then $(a,b,c,d)$ is also a strict $F_i$-switching path with respect to $(F_0,F_1,F_2)$, which is a contradiction to Lemma~\ref{lem:nostrictpaths}.
    
    Therefore we can assume that $bc$ belongs to a trivial component in $G^*$, i.e. $bc\in F_0$. Since $bc\in F'_{3-i}$, it should be the case that $bc\in S_i$, which implies that there exists an $F_i$-pentagon $(x,y,b,c,z)$ in $G$ with respect to $(F_0,F_1,F_2)$. Since $ab,cd\in F'_i\cup F'_0\subseteq F_i\cup F_0$, we know that $a,d\notin \{x,y,z\}$. Since $a,b,x,y$ and $d,c,x,z$ are alternating $F_i$-paths, we have that $ay,dz\in E(G)$. Since $a,y,z,d,a$ is an alternating 4-cycle, we know that one of $ay,dz$ is in $F_i$ and the other in $F_{3-i}$. Because of symmetry, we can assume without loss of generality that $ay\in F_i$ and $dz\in F_{3-i}$ (by renaming $(a,b,c,d)$ as $(d,c,b,a)$ and interchanging the labels of $y$ and $z$ if necessary). Then $a,y,z,x,c,d,a$ is an alternating $F_i$-circuit, implying that $ax\in  F_{3-i}$. Then $(a,x,z,d)$ is a strict $F_{3-i}$-switching path in $G$ with respect to $(F_0,F_1,F_2)$, which again contradicts Lemma~\ref{lem:nostrictpaths}.
    \hfill \qed
    \end{proof}
    
    Let $\{F,\oth{F}\}=\{F'_1,F'_2\}$. We say that $(a,b,c,d)$ is an \emph{$F$-switching cycle} in $G$ with respect to $(F'_0,F'_1,F'_2)$ if $ab,cd\in F\cup F'_0$ and $bc,ad\in\oth{F}$. As before, we say that $(a,b,c,d)$ is a \emph{switching cycle} in $G$ with respect to $(F'_0,F'_1,F'_2)$ if there exists $F\in\{F'_1,F'_2\}$ such that $(a,b,c,d)$ is an $F$-switching cycle.
    
	\begin{lemma}\label{lem:nocycle}
	There are no switching cycles in $G$ with respect to $(F'_0,F'_1,F'_2)$.
	\end{lemma}
	\begin{proof}
	Suppose not. Let $(a,b,c,d)$ be a switching cycle in $G$ with respect to $(F'_0,F'_1,F'_2)$. Let $i\in\{1,2\}$ such that $(a,b,c,d)$ is an $F'_i$-switching cycle. Then we have $ab,cd\in F'_i\cup F'_0$ and $ad,bc\in F'_{3-i}$. Suppose that $bc$ belongs to a non-trivial component of $G^*$. Then there exists $uv\in E(G)$ such that $bv,cu\in E(\overline{G})$. Since $b,c,u,v,b$ is an alternating 4-cycle and $bc\in F'_{3-i}$, we have that $uv\in F'_i$. If $u=a$ and $v=d$, then $b,(a=u),c,(d=v),b$ is an alternating 4-cycle in which both the opposite edges belong to $F'_i\cup F'_0$, which is a contradiction. Therefore, either $u\neq a$ or $v\neq d$. Because of symmetry, we can assume without loss of generality that $u\neq a$ (by renaming $(a,b,c,d)$ as $(d,c,b,a)$ and interchanging the labels of $u$ and $v$ if necessary). Then $a,b,v,u$ is an alternating $F'_i$-path, implying that $au\in E(G)$. If $au\in  F'_i\cup F'_0$ then $(c,d,a,u)$ is an $F'_i$-switching path, and if not, then $au\in F'_{3-i}$, in which case $(b,a,u,v)$ is an $F'_i$-switching path. In both cases, we have a contradiction to Lemma~\ref{lem:nopath}.

	Therefore we can assume that $bc$ belongs to a trivial component of $G^*$, i.e. $bc\in F_0$. Since $bc\in F'_{3-i}$, it should be the case that $bc\in S_i$, which implies that there exists an $F_i$-pentagon $(x,y,b,c,z)$ in $G$ with respect to $(F_0,F_1,F_2)$. Since $ab,cd\in F'_i\cup F'_0\subseteq F_i\cup F_0$, $a,d\notin \{x,y,z\}$. As $y,x,b,a$ and $z,x,c,d$ are alternating $F_i$-paths, we have that $ya,zd\in E(G)$. Now if both $ya,zd\in F'_i\cup F'_0$ we have that $(y,a,d,z)$ is an $F'_i$-switching path, which is a contradiction to Lemma~\ref{lem:nopath}. On the other hand, if $ya\in F'_{3-i}$ or $zd\in F'_{3-i}$, then since $xy,xz\in F_i\subseteq F'_i$, we have that either $(x,y,a,b)$ or $(x,z,d,c)$ is an $F'_i$-switching path, which again contradicts Lemma~\ref{lem:nopath}.
	\hfill	
	\qed
	\end{proof}
	
	We are now ready to complete the proof of Theorem~\ref{thm:main-thm}. Consider the graphs $H_1,H_2$, having $V(H_1)=V(H_2)=V(G)$, $E(H_1)=F'_1\cup F'_0$ and $E(H_2)=F'_2\cup F'_0$. We claim that $H_1$ and $H_2$ are both threshold graphs. Suppose for the sake of contradiction that $H_i$ is not a threshold graph for some $i\in\{1,2\}$. Then there exist edges $ab,cd\in E(H_i)$ such that $bc,ad\in E(\overline{H_i})$. If $bc,ad\in E(\overline{G})$, then $a,b,c,d,a$ is an alternating 4-cycle in $G$ whose opposite edges both belong to $F'_i\cup F'_0$, which contradicts the fact that $(F'_0,F'_1,F'_2)$ is a valid 3-partition. So we can assume by symmetry that $bc\in E(G)$. Since $bc\in E(\overline{H_i})$, $bc\notin F'_i\cup F'_0$, which implies that $bc\in F'_{3-i}$. Now if $ad\in E(\overline{G})$, then $(a,b,c,d)$ is an $F'_i$-switching path in $G$ with respect to $(F'_0,F'_1,F'_2)$, which is a contradiction to Lemma~\ref{lem:nopath}. On the other hand, if $ad\in E(G)$, then $ad\in F'_{3-i}$ (since $ad\in E(\overline{H_i})$), which implies that $(a,b,c,d)$ is an $F'_i$-switching cycle in $G$ with respect to $(F'_0,F'_1,F'_2)$, which contradicts Lemma~\ref{lem:nocycle}. Thus we can conclude that both $H_1$ and $H_2$ are threshold graphs. Since $E(G)=E(H_1)\cup E(H_2)$, we further get that $\{H_1,H_2\}$ is a 2-threshold cover of $G$.

    \section{Simpler proofs for paraglider-free graphs and split graphs}\label{sec:split}

    We now show that our proof of Theorem~\ref{thm:main-thm}, as well as the algorithm to construct a 2-threshold cover of a graph $G$ whose auxiliary graph $G^*$ is bipartite, becomes considerably simpler if $G$ is a ``paraglider-free'' graph or ``split graph''.
    
    A \emph{paraglider} is the graph $\overline{P_3\cup K_2}$. Note that the subgraph formed by the edges of a pentagon in a graph is a paraglider. A graph is said to be \emph{paraglider-free} if it contains no induced subgraph isomorphic to a paraglider. Thus, paraglider-free graphs cannot contain any pentagons with respect to any valid 3-partition of $E(G)$.
    
    A graph $G=(X,Y,E)$ is said to be a \emph{split graph} if $X$ is a clique in $G$, $Y$ is an independent set in $G$ and $V(G)=X\cup Y$. It is also known that split graphs are precisely $(2K_2,C_4,C_5)$-free graphs. As the paraglider contains an induced $C_4$, split graphs are paraglider-free.
    
    Let $G$ be a graph such that $G^*$ is bipartite.
    Suppose that $G$ is paraglider-free. Then the proof of Theorem~\ref{thm:main-thm} can be simplified as follows. We skip Phase~III of our algorithm. Thus, once finish running Phases~I and~II of our algorithm and obtain the valid 3-partition $(F_0,F_1,F_2)$ of $E(G)$, we output $H_1=(V(G),F_1\cup F_0)$ and $H_2=(V(G),F_2\cup F_0)$ as the two threshold graphs that form a 2-threshold cover of $G$. We can do this because, the fact that $G$ is paraglider-free implies that $G$ does not contain any pentagons after Phase~II. Thus we can conclude that Lemma~\ref{lem:nostrictpent} holds without any proof (the whole of Section~\ref{sec:nostrictpent} can be omitted). Using Lemma~\ref{lem:nostrictpent}, we can prove Observations~\ref{obs:pathandpent},~\ref{type1} and~\ref{type2} as before without any modification. We now simply set $F'_0=F_0$, $F'_1=F_1$ and $F'_2=F_2$ without running Phase~III, as the fact that there are no pentagons in $G$ implies that $S_1=S_2=\emptyset$. The statement of Lemma~\ref{lem:nopentagons} can be directly seen to be true without any proof. Lemmas~\ref{lem:nopath} and~\ref{lem:nocycle} can be proved as before; actually, the second paragraphs of both these proofs can be omitted as these cases only arise when $bc\in S_i$ for some $i\in\{1,2\}$. It now follows as before that $H_1$ and $H_2$ form a 2-threshold cover of $G$.

    For the case of split graphs (which are a special kind of paraglider-free graphs), we can additionally also skip Phase~I of our algorithm. Suppose that $G=(X,Y,E)$ is a split graph. We start with an arbitrary ordering $<$ of the vertices of $G$, and once we get the valid 3-partition $(F_0,F_1,F_2)$ after running Phase~II of the algorithm, we can output $H_1=(V(G),F_1\cup F_0)$ and $H_2=(V(G),F_2\cup F_0)$ as the two threshold graphs that form a 2-threshold cover of $G$. We follow the same proof as the one for paraglider-free graphs, with the only change being made to the last paragraph of the proof of Observation~\ref{type2}, where Observation~\ref{obs:lexbfs} is used (note that Observation~\ref{obs:lexbfs} no longer holds as $<$ is not necessarily a Lex-BFS ordering). We replace this paragraph with the following:
    
	\begin{quotation}
    Recall that $a_0b_0,a_1b_1,\ldots,a_kb_k$ is a path in $G^*$, such that for any  $i\in \{0,1,\ldots,k-1\}$, $a_ib_{i+1}\in E(\overline{G})$ and $b_ia_{i+1}\in E(\overline{G})$.
    Let $i\in\{0,1,\ldots,k-1\}$. If $a_i$ and $b_{i+1}$ both belong to one of $X$ or $Y$, then it should be the case that $a_i,b_{i+1}\in Y$ (recall that $X$ is a clique in $G$). Since $a_ib_i,a_{i+1}b_{i+1}\in E(G)$ and $Y$ is an independent set in $G$, we then have $b_i,a_{i+1}\in X$. Since $X$ is a clique, this contradicts the fact that $b_ia_{i+1}\in E(\overline{G})$. Therefore we can conclude that for each $i\in\{0,1,\ldots,k-1\}$, one of $a_i,b_{i+1}$ belongs to $X$ and the other to $Y$.
    By the same argument, we can also show that for each $i\in\{0,1,\ldots,k-1\}$, one of $b_i,a_{i+1}$ belongs to $X$ 
    and the other to $Y$. Since $k$ is odd, it now follows that one of $(a=a_0),b_k$ belongs to $X$ and the other to $Y$, and similarly, one of $(b=b_0),a_k$ belongs to $X$ and the other to $Y$. We can therefore conclude that $a\neq b_k$ and $b\neq a_k$. Recall that $a_kb_k<ab$, $a<d$, $(c,b_k,a_k,u)$ is a strict $F_2$-switching path, and $a_iu,b_ic\in F_1$ (resp. $a_iu,b_ic\in F_2$) for each even $i$ (resp. odd $i$). Then we have $a_ku,b_kc\in F_2$ and $au,bc\in F_1$, which implies that $a_k\neq a$ and $b_k\neq b$. We now have $\{a,b\}\cap\{a_k,b_k\}=\emptyset$, and therefore $\min\{a_k,b_k\}<\min\{a,b\}$. But then as $a<d$, we have $\{c,b_k,a_k,u\}<\{a,b,c,d\}$, which is a contradiction to the choice of $(a,b,c,d)$.
	\end{quotation}
     
    Ibaraki and Peled~\cite{ibarakipeled} were the first to show that if $G$ is a split graph, then $G$ has a 2-threshold cover if and only if $G^*$ is bipartite. Our proof, simplified as described above, yields a different proof for this fact which we believe is much simpler than the proofs in~\cite{ibarakipeled} or~\cite{raschle1995recognition}.

	\section{Conclusion}    
	\noindent\textbf{The Chain Subgraph Cover Problem.}
	A bipartite graph $G=(A,B,E)$ is called a \emph{chain graph} if it does not contain a pair of edges whose endpoints induce a $2K_2$ in $G$. Let $\hat{G}$ be the split graph obtained from $G$ by adding edges between every pair of vertices in $A$ (or $B$). It can be seen that $G$ is a chain graph if and only if $\hat{G}$ is a threshold graph. A collection of chain graphs $\{H_1,H_2,\ldots,H_k\}$ is said to be a \emph{$k$-chain subgraph cover} of a bipartite graph $G$ if it is covered by $H_1,H_2,\ldots,H_k$. Yannakakis~\cite{yannakakis1982complexity} credits Martin Golumbic for observing that a bipartite graph $G$ has a $k$-chain subgraph cover if and only if $\hat{G}$ has a $k$-threshold cover. The problem of deciding whether a bipartite graph $G$ can be covered by $k$ chain graphs, i.e. whether $G$ has a $k$-chain subgraph cover, is known as the \emph{$k$-chain subgraph cover ($k$-CSC)} problem. Yannakakis~\cite{yannakakis1982complexity} showed that 3-CSC is NP-complete, which implies that the problem of deciding whether $\thd(G)\leq 3$ for an input graph $G$ is also NP-complete. He also pointed out that using the Golumbic's observation and the results of Ibaraki and Peled~\cite{ibarakipeled}, the $2$-CSC problem can be solved in polynomial time, as it can be reduced to the problem of determining whether a split graph can be covered by two threshold graphs.
	Thus our algorithm for split graphs described in Section~\ref{sec:split} can also be used to compute a 2-chain subgraph cover, if one exists, for an input bipartite graph $G$ in time $O(|E(G)|^2)$ (note that even though $|E(\hat{G})|>|E(G)|$, the vertices in $\hat{G}^*$ corresponding to the edges in $E(\hat{G})\setminus E(G)$ are all isolated vertices and hence can be ignored while computing the partial 2-coloring of $\hat{G}^*$). Note that Ma and Spinrad~\cite{ma19942} propose a more involved $O(|V(G)|^2)$ algorithm for the problem. However, our algorithm for split graphs, and hence the algorithm for computing a 2-chain subgraph cover that it yields, is considerably simpler to implement than the algorithms of~\cite{ibarakipeled,ma19942,raschle1995recognition,sterbini1998n3}.
	
	
	\medskip
	
	Would running just Phases~II and~III of our algorithm always produce a valid 2-threshold cover of $G$ for any graph $G$? That is, could we have started with an arbitrary ordering of the vertices of $G$ instead of a Lex-BFS ordering? We show that the algorithm may fail to produce a 2-threshold cover of the graph $G$ shown in Figure~\ref{fig:auxgraph} if the algorithm starts by taking an arbitrary ordering of vertices in Phase~I. Suppose that the vertices of the graph are ordered according to their labels as shown in Figure~\ref{fig:bad}(a). Clearly, it is not a Lex-BFS ordering of the vertices, as since the vertex in the second position is not a neighbor of the vertex in the first position, it is not even a BFS ordering. The sets $F'_0,F'_1,F'_2$ computed by our algorithm after Phases~II and~III will be as shown in Figure~\ref{fig:bad}(b)---the vertices of $G^*$ in the set $F'_1$ are shown as black, the ones in $F'_2$ as gray and the ones in $F'_0$ as white. In Figure~\ref{fig:bad}(a), the black edges form the graph $H_1$ and the gray edges form the graph $H_2$. Clearly, neither is a threshold graph (for example, both contain a $C_4$). On the other hand, Figure~\ref{fig:lexbfs} shows the 2-threshold cover of $G$ computed by our algorithm if it starts with the Lex-BFS ordering of the vertices of $G$ as indicated by the labels of the vertices in Figure~\ref{fig:lexbfs}(a). Note that starting with a BFS ordering instead of a Lex-BFS ordering will also not work, since we can always add a universal vertex to the graph $G$ shown in Figure~\ref{fig:bad}(a) and number it 0, so that the vertex ordering is now a BFS ordering. It is not difficult to see that the graphs $H_1$ and $H_2$ computed in this case also fail to be threshold graphs (in fact, the edges incident on the vertex labelled 0 are all isolated vertices in the auxiliary graph, and none of them belong to any pentagons; hence they all belong to $F'_0$, and the sets $F'_1$ and $F'_2$ will be exactly the same as before).
	
	\begin{figure}
	\begin{tabular}{p{.49\textwidth}p{.49\textwidth}}
	\parbox{.49\textwidth}{
	\centering
	\renewcommand{\defradius}{0.1}
	\renewcommand{\vertexset}{(6,2,2),(7,0,3),(2,2,3),(4,0,1),(3,4,3),(5,4,1),(1,2,0)}
	\renewcommand{\edgeset}{(1,4,black,2),(1,5,lightgray,2),(2,3,lightgray,2),(2,4,black,2),(2,5,lightgray,2),(2,7,black,2),(3,5,black,2),(3,6,lightgray,2),(4,6,black,2),(4,7,lightgray,2),(5,6,lightgray,2),(6,7,black,2)}
	\begin{tikzpicture}
	\draw [black,line width=2] (0,1.04)--(4,1.04);
	\draw [lightgray,line width=2] (0,0.96)--(4,0.96);
	\drawgraph
	\node[below=2] at (\xy{1}) {1};
	\node[above=2] at (\xy{2}) {2};
	\node[above=2] at (\xy{3}) {3};
	\node[below=2] at (\xy{4}) {4};
	\node[below=2] at (\xy{5}) {5};
	\node[below=2] at (\xy{6}) {6};
	\node[above=2] at (\xy{7}) {7};
	\end{tikzpicture}
	}&
	\parbox{.49\textwidth}{
	\centering
	\renewcommand{\vertexset}{(14,2,0,black),(15,2,1,lightgray),(23,1,0,lightgray),(24,4,0,black),(25,0,1,lightgray),(27,3,1,black),(35,2.5,2,black),(36,3,0,lightgray),(45,2,3),(46,0,0,black),(47,1.5,2,lightgray),(56,4,1,lightgray),(67,1,1,black)}
	\renewcommand{\edgeset}{(14,23),(14,36),(15,27),(15,67),(23,46),(23,67),(24,36),(25,67),(27,36),(27,56),(35,47)}
	\begin{tikzpicture}
	\drawgraph
	\node at (\xy{14}) {\textcolor{white}{14}};
	\node at (\xy{15}) {15};
	\node at (\xy{23}) {23};
	\node at (\xy{24}) {\textcolor{white}{24}};
	\node at (\xy{25}) {25};
	\node at (\xy{27}) {\textcolor{white}{27}};
	\node at (\xy{35}) {\textcolor{white}{35}};
	\node at (\xy{36}) {36};
	\node at (\xy{45}) {45};
	\node at (\xy{46}) {\textcolor{white}{46}};
	\node at (\xy{47}) {47};
	\node at (\xy{56}) {56};
	\node at (\xy{67}) {\textcolor{white}{67}};
	\end{tikzpicture}
	}\vspace{.1in}\\
	\centering (a)&\centering (b)
	\end{tabular}
	\caption{(a) The graph $G$ from Figure~\ref{fig:auxgraph}, with its vertices numbered according to a non-Lex-BFS ordering, and (b) the graph $G^*$ and its partial 2-coloring after Phases~II and~III.}\label{fig:bad}
	\end{figure}
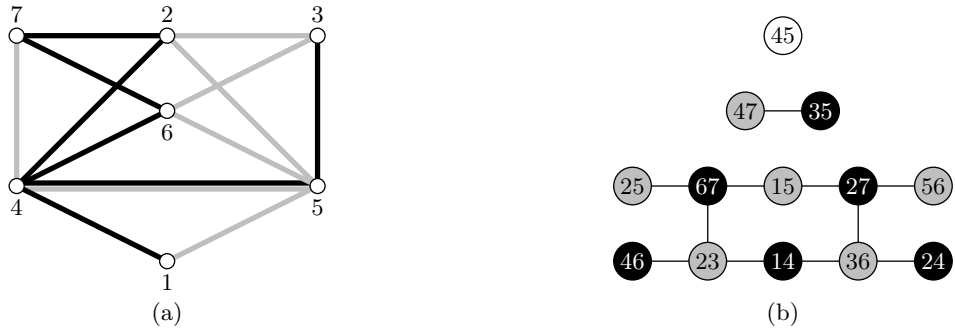

	\begin{figure}
	\begin{tabular}{p{.49\textwidth}p{.49\textwidth}}
	\parbox{.49\textwidth}{
	\centering
	\renewcommand{\defradius}{0.1}
	\renewcommand{\vertexset}{(6,2,2),(7,0,3),(2,2,3),(4,0,1),(3,4,3),(5,4,1),(1,2,0)}
	\renewcommand{\edgeset}{(1,4,black,2),(1,5,lightgray,2),(2,3,lightgray,2),(2,4,black,2),(2,5,lightgray,2),(2,7,black,2),(3,5,lightgray,2),(3,6,lightgray,2),(4,6,black,2),(4,7,black,2),(5,6,lightgray,2),(6,7,black,2)}
	\begin{tikzpicture}
	\draw [black,line width=2] (0,1.04)--(4,1.04);
	\draw [lightgray,line width=2] (0,0.96)--(4,0.96);
	\drawgraph
	\node[below=2] at (\xy{1}) {1};
	\node[above=2] at (\xy{2}) {5};
	\node[above=2] at (\xy{3}) {7};
	\node[below=2] at (\xy{4}) {2};
	\node[below=2] at (\xy{5}) {3};
	\node[below=2] at (\xy{6}) {4};
	\node[above=2] at (\xy{7}) {6};
	\end{tikzpicture}
	}&
	\parbox{.49\textwidth}{
	\centering
	\renewcommand{\vertexset}{(14,2,0,black),(15,2,1,lightgray),(23,1,0,lightgray),(24,4,0,black),(25,0,1,lightgray),(27,3,1,black),(35,2.5,2,lightgray),(36,3,0,lightgray),(45,2,3),(46,0,0,black),(47,1.5,2,black),(56,4,1,lightgray),(67,1,1,black)}
	\renewcommand{\edgeset}{(14,23),(14,36),(15,27),(15,67),(23,46),(23,67),(24,36),(25,67),(27,36),(27,56),(35,47)}
	\begin{tikzpicture}
	\drawgraph
	\node at (\xy{14}) {\textcolor{white}{12}};
	\node at (\xy{15}) {13};
	\node at (\xy{23}) {57};
	\node at (\xy{24}) {\textcolor{white}{25}};
	\node at (\xy{25}) {35};
	\node at (\xy{27}) {\textcolor{white}{56}};
	\node at (\xy{35}) {37};
	\node at (\xy{36}) {47};
	\node at (\xy{45}) {23};
	\node at (\xy{46}) {\textcolor{white}{24}};
	\node at (\xy{47}) {\textcolor{white}{26}};
	\node at (\xy{56}) {34};
	\node at (\xy{67}) {\textcolor{white}{46}};
	\end{tikzpicture}
	}\vspace{.1in}\\
	\centering (a)&\centering (b)
	\end{tabular}
	\caption{(a) The graph $G$ from Figure~\ref{fig:auxgraph}, with its vertices numbered according to a Lex-BFS ordering, and (b) the graph $G^*$ and its partial 2-coloring after Phases~II and~III.}\label{fig:lexbfs}
	\end{figure}
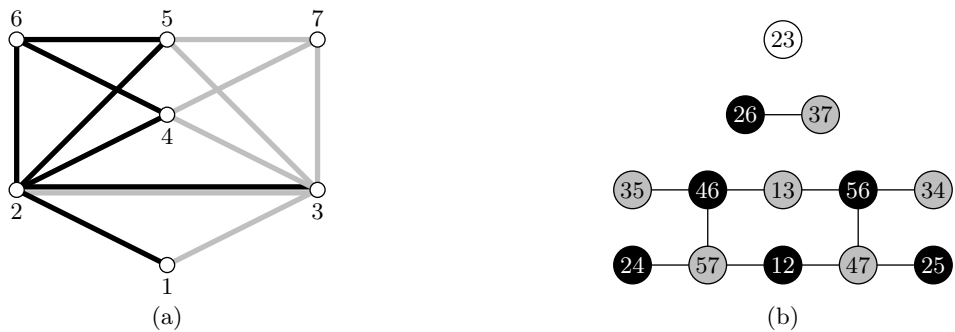
		
	Thus the graph $G$ shown in Figure~\ref{fig:auxgraph} demonstrates that even though Phase~I is optional for split graphs, for general graphs, our algorithm may not produce a 2-threshold cover of the input graph if Phase~I is skipped. Note that the graph $G$ is not a paraglider-free graph. We have not found an example of a paraglider-free graph for which our algorithm will fail if Phase~I is skipped.
	\section*{Acknowledgements}
	The authors wish to express their gratitude to Prof. Nadimpalli Mahadev, who pointed out a serious error in an earlier version of this paper and also to Prof. Martin Golumbic, for his helpful pointers.
	\bibliographystyle{splncs04}
	\bibliography{reference} 
\end{document}